\documentclass[10pt,a4paper]{article}

\usepackage{times}
\usepackage{amssymb,amsmath,amsfonts}
\usepackage{amsthm}
\usepackage{latexsym}
\usepackage{verbatim}

\usepackage{enumerate}
\usepackage{graphicx}
\usepackage{amssymb}
\usepackage{tikz}

\newcommand{\F}{\ensuremath{\mathcal{F}}}

\newcommand{\Prof}{\ensuremath{\mathcal{P}}}

\newcommand{\Leaves}{\ensuremath{\mathcal{L}}}

\newtheorem{theorem}{Theorem}
\newtheorem{lemma}{Lemma}
\newtheorem{observation}{Observation}

\theoremstyle{remark}

\theoremstyle{definition}
\newtheorem{example}{Example}

\newtheoremstyle{dotless}{}{}{\itshape}{}{\bfseries}{}{ }{}
\theoremstyle{dotless}

\begin{document}
\title{On Two Graph-Theoretic Characterizations of Tree Compatibility\thanks{This work was supported in part by the National Science Foundation under grants CCF-1017189 and DEB-0829674.}}
\author{Sudheer Vakati\footnote{Department of Computer Science, Iowa State University, Ames, IA 50011} \and David Fern\'{a}ndez-Baca\footnotemark[\value{footnote}]}
\date{}
\maketitle

\begin{abstract}
Deciding whether a collection of unrooted trees is compatible is a fundamental problem in phylogenetics. Two different graph-theoretic characterizations of tree compatibility have recently been proposed. In one of these, tree compatibility is characterized in terms of the existence of a specific kind of triangulation in a structure known as the display graph. An alternative characterization expresses the tree compatibility problem as a chordal graph sandwich problem in a structure known as the edge label intersection graph. In this paper we show that the characterization using edge label intersection graphs transforms to a characterization in terms of minimal cuts of the display graph. We show how these two characterizations are related to compatibility of splits. We also show how the characterization in terms of minimal cuts of display graph is related to the characterization in terms of triangulation of the display graph.
\end{abstract}

\section{Introduction}
A \emph{phylogenetic tree} $T$ is an unrooted tree whose leaves are bijectively mapped to label set $\Leaves(T)$. Labels represent species and phylogenetic trees represent evolutionary history of species. Let $\Prof$ be a collection of phylogenetic trees. We call $\Prof$ a \emph{profile}. We denote $\bigcup_{T \in \Prof}\Leaves(T)$ by $\Leaves(\Prof)$. A \emph{supertree} of $\Prof$ is a phylogenetic tree whose label set is $\Leaves(\Prof)$.

Let $S$ be a phylogenetic tree. For any $Y \subseteq \Leaves(S)$ let $S|Y$ denote the tree obtained by suppressing any degree two vertices in the minimal subtree of $S$ connecting the labels in $Y$. Let $T$ be a phylogenetic tree where $\Leaves(T) \subseteq \Leaves(S)$. We say that $S$ \emph{displays} $T$ if $T$ can be derived from $S|\Leaves(T)$ by contraction of edges. Given a profile $\Prof$ of trees, the \emph{tree compatibility} problem asks if there exists a supertree of $\Prof$ which displays all the trees in $\Prof$. If such a supertree $S$ exists, we say that $\Prof$ is compatible and $S$ is a compatible tree of $\Prof$. The tree compatibility problem is NP-complete~\cite{Steel92} but fixed parameter tractable when parametrized by number of trees~\cite{BryantLagergren06}.

An instance of the tree compatibility problem can be transformed to an instance of character compatibility by representing each input tree by its set of splits~\cite{SempleSteel03}. A \emph{quartet} is a binary phylogenetic tree with exactly four leaves. An instance of tree compatibility can also be transformed into an instance of quartet compatibility~\cite{Steel92}. Grunewald et al.~\cite{Grunewald2008} characterized quartet compatibility problem in terms of finding an unification sequence for a structure called the quartet graph.

Vakati and Fern\'{a}ndez-Baca characterized the tree compatibility problem in terms of finding a legal triangulation~\cite{Vakati11} of the display graph of a profile, a graph introduced by Bryant and Lagergren~\cite{BryantLagergren06}. Gysel et al. introduced the edge label intersection graph for a profile of phylogenetic trees and used this graph to characterize tree compatibility as a chordal sandwich problem~\cite{Gysel2012}. Here we show that the latter characterization translates to a characterization in terms of minimal cuts of display graphs. We also show how such cuts are closely related to the splits of the compatible supertree. Finally, we show how these two characterizations relate to the legal triangulation characterization given in~\cite{Vakati11}.

\section{Preliminaries}

 For every nonnegative integer $m$, we denote the set $\{1, \dots, m\}$ by $[m]$. Let $G$ be a graph. We represent the vertices and edges of $G$ by $V(G)$ and $E(G)$ respectively. For any $U \subseteq V(G)$, $G-U$ represents the graph derived by removing vertices of $U$ and their incident edges from $G$. Similarly, for any $F \subseteq E(G)$, $G-F$ represents the graph with vertex set $V(G)$ and edge set $E(G) \setminus F$. For any vertex $v \in V(G)$, we denote the set $\{x: \{x, v\} \in E(G)\}$ by $N_G(v)$.

For any two nonadjacent vertices $a$ and $b$ of $G$, an $a$-$b$ $separator$ $U$ of $G$ is a set of vertices whose removal disconnects $G$ and $a$ and $b$ are in different connected components of $G-U$. An $a$-$b$ separator $U$ is \emph{minimal} if for any $U' \subset U$, $U'$ is not an $a$-$b$ separator. A set $U \subseteq V(G)$ is a \emph{minimal separator} if $U$ is a minimal $a$-$b$ separator for some nonadjacent vertices $a$ and $b$ of $G$. Two minimal separators $U$ and $U'$ are \emph{parallel} if $G-U$ contains at most one component $H$ where $V(H) \cap U' \neq \emptyset$. We represent the set of all minimal separators of graph $G$ by $\triangle_G$.

Assume that $G$ is connected. A \emph{cut} is a set of edges $F \subseteq E(G)$ whose removal disconnects $G$. A cut $F$ is \emph{minimal} if there does not exist $F' \subset F$ where  $G-F'$ is disconnected. Note that if $F$ is minimal, there will exactly be two connected components in $G-F$. Two minimal cuts $F$ and $F'$ are \emph{parallel} if $G-F$ has at most one connected component $H$ where $E(H) \cap F' \neq \emptyset$. 

A \emph{chord} is an edge between two nonadjacent vertices of a cycle. A graph $H$ is \emph{chordal} if and only if every cycle of length four or greater in $H$ has a chord. A chordal graph $H$ is a \emph{triangulation} of graph $G$ if and only if $V(G) = V(H)$ and $E(G) \subseteq E(H)$. The edges in $E(H) \setminus E(G)$ are called \emph{fill-in} edges of $G$. A \emph{clique tree} of a chordal graph $H$ is a pair $(T, B)$ where (i) $T$ is a tree, (ii) $B$ is a bijective function from vertices of $T$ to maximal cliques of $H$, and (iii) $(T, B)$ satisfies the \emph{coherence} property; i.e., for every vertex $v \in H$, the set of all vertices $x$ of $T$ where $v \in B(x)$ induces a subtree in $T$.

Let $G$ be a graph and let $\F$ be a collection of subsets of $V(G)$. We represent by $G_{\F}$ the graph derived from $G$ by making the set of vertices of $X$ a clique in $G$ for every $X \in \F$. 

\begin{theorem}[\cite{Bouchitte2001,Heggernes2006,Parra1997}]
\label{thm:parallel_minimal_ct}
Let $\F$ be a maximal set of pairwise parallel minimal separators of $G$ and let $H$ be a minimal triangulation of $G$.
\begin{enumerate}
\item $G_{\F}$ is a minimal triangulation of $G$.
\item Let $(T, B)$ be a clique tree of $G_{\F}$. There exists a minimal separator $F \in \F$ if and only if there exist two adjacent vertices $x$ and $y$ in $T$ where $B(x) \cap B(y) = F$.
\item $\triangle_{H}$ is a maximal set of pairwise parallel minimal separators of $G$ and $G_{\triangle_H} = H$.
\end{enumerate}
\end{theorem}

 Let $T$ be a phylogenetic tree over label set $\Leaves(T)$. Since there exists a bijective function from leaves of $T$ to $\Leaves(T)$, we will represent leaves of $T$ by their labels.  Let $\Prof = \{T_1, T_2, \cdots, T_k\}$ be a profile of $k$ phylogenetic trees. The \emph{display graph} of profile $\Prof$, denoted by $G(\Prof)$, is a graph whose vertex set is $\bigcup_{i \in [k]} V(T_i)$ and edge set is $\bigcup_{j \in [k]} E(T_j)$. An example of a display graph is given in Fig.~\ref{fig:example}. A vertex $v$ of $G(\Prof)$ is a \emph{leaf} vertex if $v \in \Leaves(\Prof)$. Every other vertex of $G(\Prof)$ is an \emph{internal} vertex. An edge of $G(\Prof)$ is an \emph{internal} edge if both its endpoints are internal; otherwise, it is a \emph{non-internal} edge. Let $H$ be a subgraph of $G(\Prof)$. We represent by $\Leaves(H)$, the set of all leaf vertices of $H$. A triangulation $G'$ of $G(\Prof)$ is \emph{legal} if it satisfies the following conditions.

\begin{enumerate} [(LT1)]
\item For every clique $C$ of $G'$, if $C$ contains an internal edge, then it cannot contain any other edge of $G(\Prof)$.  
\item There does not exist a fill-in edge with a leaf vertex as an endpoint.
\end{enumerate}

\begin{theorem}[Vakati, Fern\'{a}ndez-Baca~\cite{Vakati11}]
\label{thm:lt}
A profile $\Prof$ of unrooted phylogenetic trees is compatible if and only if $G(\Prof)$ has a legal triangulation.
\end{theorem}

The \emph{edge label intersection} graph (see Fig.~\ref{fig:example}) of a profile $\Prof$, denoted $L(\Prof)$, is the graph whose vertex set is the set of all edges of input trees of $\Prof$ where there is an edge between vertices $e$ and $e'$ if and only if $e \cap e' \neq \emptyset$~\cite{Gysel2012}. (Note that in~\cite{Gysel2012}, $L(\Prof)$ is the modified edge label intersection graph.) 

It can be verified that $L(\Prof)$ is the line graph of $G(\Prof)$~\cite{Gysel2012}; i.e., the vertices of $L(\Prof)$ are the edges of $G(\Prof)$ and two vertices in $L(\Prof)$ are adjacent if the corresponding edges in $G(\Prof)$ share a common endpoint. 

A fill-in edge of $L(\Prof)$ is \emph{valid} if both its endpoints are not from $L(T)$ for some input tree $T$. A triangulation $H$ of $L(\Prof)$ is \emph{restricted} if every fill-in edge of $H$ is valid.

\begin{theorem}[Gysel, Stevens, Gusfield~\cite{Gysel2012}]
\label{thm:EL_compatibility}
A Profile $\Prof$ of unrooted phylogenetic trees is compatible if and only if there exists a restricted triangulation of $L(\Prof)$.
\end{theorem}

For rest of the paper we will assume that for any profile $\Prof$, $G(\Prof)$ is connected. Otherwise, there exists a partition $\Prof_{part}$ of $\Prof$ such that for every $P \in \Prof_{part}$, $G(P)$ is connected and $P$ is maximal such set. Then, $\Prof$ is compatible if and only if $P$ is compatible for every $P \in \Prof_{part}$. Note that, if $G(\Prof)$ is connected, then $L(\Prof)$ is also connected.
 
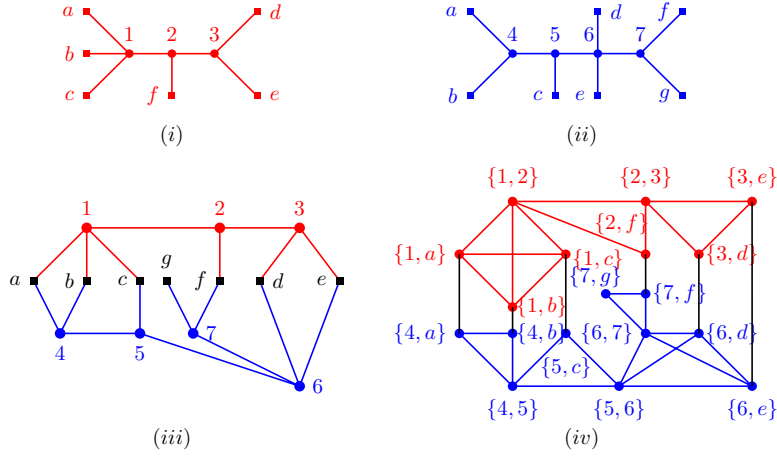
\begin{figure}[t!]
\renewcommand{\figurename}{Fig.}
 \begin{center}
\begin{tikzpicture} [inner sep = 2mm, semithick,scale=0.7]
\tikzstyle{every node}=[font=\fontsize{10}{10}, scale=0.8];

\coordinate (one_a) at (1,9);
\begin{scope}[red,scale=0.8]
\draw (one_a) --++(1, -1) coordinate(one_1) --++(-1, 0) coordinate(one_b) ++(1,0) --++(-1,-1) coordinate(one_c)  ++(1,1) --++(1,0) coordinate(one_2) --++(0,-1) coordinate(one_f) ++(0,1) --++(1,0) coordinate (one_3) --++(1,1) coordinate(one_d) ++(-1,-1) --+(1,-1) coordinate(one_e);
\path [fill] (one_1) circle(0.1) node[above] {$1$};
\path [fill] (one_2) circle(0.1) node[above] {$2$};
\path [fill] (one_3) circle(0.1) node[above] {$3$};
\path [fill] (one_a) node[left] {$a$} ++(-0.075,-0.075)  rectangle +(0.15,0.15);
\path [fill] (one_b) node[left] {$b$} ++(-0.075,-0.075)  rectangle +(0.15,0.15);
\path [fill] (one_c) node[left] {$c$} ++(-0.075,-0.075)  rectangle +(0.15,0.15);
\path [fill] (one_d) node[right]{$d$}++(-0.075,-0.075)  rectangle +(0.15,0.15);
\path [fill] (one_e) node[right]{$e$}++(-0.075,-0.075)  rectangle +(0.15,0.15);
\path [fill] (one_f) node[left]{$f$}++(-0.075,-0.075)  rectangle +(0.15,0.15);
\end{scope}

\path(one_f) ++(0, -0.75) coordinate(i)node{$(i)$};
\path (one_d) ++(4,0) coordinate(two_a);
\begin{scope}[blue,scale=0.8]
\draw (two_a) --++(1,-1) coordinate(two_4) --++(-1,-1) coordinate(two_b) ++(1,1) --++(1,0) coordinate(two_5)--++(0,-1) coordinate(two_c) ++(0,1) --++(1,0) coordinate(two_6) --++(0,1) coordinate(two_d) ++(0,-2)  coordinate(two_e)--++(0, 1) --++(1,0) coordinate(two_7) --++(1,1) coordinate(two_f) ++(-1,-1) --+(1,-1) coordinate(two_g);
\begin{scope}
\path [fill] (two_4) circle(0.1) node[above] {$4$};
\path [fill] (two_5) circle(0.1) node[above] {$5$};
\path [fill] (two_6) circle(0.1) +(-0.2,0) node[above] {$6$};
\path [fill] (two_7) circle(0.1) node[above] {$7$};
\end{scope}
\path [fill] (two_a) node[left] {$a$} ++(-0.075,-0.075)  rectangle +(0.15,0.15);
\path [fill] (two_b) node[left] {$b$} ++(-0.075,-0.075)  rectangle +(0.15,0.15);
\path [fill] (two_c) node[left] {$c$} ++(-0.075,-0.075)  rectangle +(0.15,0.15);
\path [fill] (two_d) node[right] {$d$} ++(-0.075,-0.075)  rectangle +(0.15,0.15);
\path [fill] (two_e) node[left] {$e$} ++(-0.075,-0.075)  rectangle +(0.15,0.15);
\path [fill] (two_f) node[left] {$f$} ++(-0.075,-0.075)  rectangle +(0.15,0.15);
\path [fill] (two_g) node[left] {$g$} ++(-0.075,-0.075)  rectangle +(0.15,0.15);
\end{scope}
\path (two_c) ++(0.5, -0.75) coordinate(ii)node{$(ii)$};
\path (one_c) ++(0,-2.5) coordinate(dis_1);

\begin{scope}[red]
\draw (dis_1) --++(-1, -1) coordinate(dis_a)  ++(1, 0) coordinate(dis_b) --++(0,1)--++(1,-1) coordinate(dis_c) ++(-1,1) --++(2.5,0) coordinate(dis_2) --++(0,-1) ++(0,1) --++(1.5,0) coordinate(dis_3)--++(-0.75,-1)coordinate(dis_d) ++(1.5,0) coordinate(dis_e) -- (dis_3);
\begin{scope}
\path [fill] (dis_1) circle(0.1) node[above] {$1$};
\path [fill] (dis_2) circle(0.1) node[above] {$2$};
\path [fill] (dis_3) circle(0.1) node[above] {$3$};
\end{scope}

\end{scope}

\begin{scope}[blue]
\draw (dis_a) --++(0.5,-1) coordinate(dis_4)--++(0.5,1) ++(-0.5,-1) --++(1.5,0) coordinate(dis_5)--++(0,1) ++(0,-1)--++(3,-1) coordinate(dis_6) --++(-2,1) coordinate(dis_7) --++(-0.5,1)coordinate(dis_g) ++(1,0) coordinate(dis_f) --++(-0.5,-1) (dis_6)--(dis_d) (dis_6)--(dis_e);

\begin{scope}
\path [fill] (dis_4) circle(0.1) node[below] {$4$};
\path [fill] (dis_5) circle(0.1) node[below] {$5$};
\path [fill] (dis_6) circle(0.1) node[right] {$6$};
\path [fill] (dis_7) circle(0.1) node[right] {$7$};
\end{scope}

\end{scope}

\begin{scope}[black]
\path [fill] (dis_a) node[left] {$a$} ++(-0.075,-0.075)  rectangle +(0.15,0.15);
\path [fill] (dis_b) node[left] {$b$} ++(-0.075,-0.075)  rectangle +(0.15,0.15);
\path [fill] (dis_c) node[left] {$c$} ++(-0.075,-0.075)  rectangle +(0.15,0.15);
\path [fill] (dis_d) node[right] {$d$} ++(-0.075,-0.075)  rectangle +(0.15,0.15);
\path [fill] (dis_e) node[left] {$e$} ++(-0.075,-0.075)  rectangle +(0.15,0.15);
\path [fill] (dis_f) node[left] {$f$} ++(-0.075,-0.075)  rectangle +(0.15,0.15);
\path [fill] (dis_g) node[above] {$g$} ++(-0.075,-0.075)  rectangle +(0.15,0.15);
\end{scope}

\path (i) ++(0, -5.75) node{$(iii)$};
\path (dis_3) ++(4,0.5) coordinate(el_12);

\begin{scope}[red]
\draw (el_12)--++(-1, -1) coordinate(el_1a) --++(1,-1) coordinate(el_1b) --++(1, 1) coordinate(el_1c) --(el_12) --(el_1b) (el_1a)--(el_1c);
\draw (el_12)--++(2.5,-1) coordinate(el_2f)--++(0,1) coordinate(el_23) -- (el_12);
\draw (el_23)--++(1,-1) coordinate(el_3d)--++(1,1) coordinate(el_3e)--(el_23);

\begin{scope}
\path [fill] (el_12) circle(0.09) node[above]{$\{1,2\}$};
\path [fill] (el_1a) circle(0.09) node[left]{$\{1,a\}$};
\path [fill] (el_1c) circle(0.09) ++(0,-0.1)node[right, inner sep=1mm]{$\{1,c\}$};
\path [fill] (el_1b) circle(0.09) node[right, inner sep=0.5mm]{$\{1,b\}$};
\path [fill] (el_23) circle(0.09) node[above]{$\{2,3\}$};
\path [fill] (el_2f) circle(0.09) ++(-0.45,0.3)node[above, inner sep=0.5mm]{$\{2,f\}$};
\path [fill] (el_3e) circle(0.09) node[above]{$\{3,e\}$};
\path [fill] (el_3d) circle(0.09) node[right, inner sep=1mm]{$\{3,d\}$};
\end{scope}

\end{scope}

\path (el_1a) ++(0,-1.5) coordinate(el_4a);
\begin{scope}[blue]
\draw (el_4a) --++(1,0)coordinate(el_4b)--++(0,-1) coordinate(el_45) --(el_4a);
\draw (el_45)--++(1,1)coordinate(el_5c)--++(1,-1) coordinate(el_56) --(el_45);
\draw (el_56)--++(0.5,1)coordinate(el_67)--++(1,0) coordinate(el_6d)--++(1,-1)coordinate(el_6e) -- (el_56) -- (el_6d) (el_67)--(el_6e);
\draw (el_67) --++(0, 0.75)coordinate(el_7f)--++(-0.75,0) coordinate(el_7g)--(el_67);

\begin{scope}
\path [fill] (el_4a) circle(0.09) node[left]{$\{4,a\}$};
\path [fill] (el_4b) circle(0.09) node[right, inner sep=0.5mm]{$\{4,b\}$};
\path [fill] (el_45) circle(0.09) node[below]{$\{4,5\}$};
\path [fill] (el_5c) circle(0.09) ++(0,-0.2)node[below]{$\{5,c\}$};
\path [fill] (el_56) circle(0.09) node[below]{$\{5,6\}$};
\path [fill] (el_67) circle(0.09) node[left]{$\{6,7\}$};
\path [fill] (el_7g) circle(0.09) ++(-0.2,0)node[above, inner sep=1mm]{$\{7,g\}$};
\path [fill] (el_7f) circle(0.09) node[right,inner sep=1mm]{$\{7,f\}$};
\path [fill] (el_6d) circle(0.09) node[right, inner sep=1mm]{$\{6,d\}$};
\path [fill] (el_6e) circle(0.09) node[below]{$\{6,e\}$};
\end{scope}
\end{scope}

\begin{scope}[black]
\draw (el_1a)--(el_4a);
\draw (el_1b)--(el_4b);
\draw (el_1c)--(el_5c);
\draw (el_3d)--(el_6d);
\draw (el_3e)--(el_6e);
\draw (el_2f)--(el_7f);
\end{scope}

\path (ii) ++(0,-5.75) node{$(iv)$};

\end{tikzpicture}
 \end{center}
\caption{$(i)$ First input tree. $(ii)$ Second input tree, which is compatible with the first. $(iii)$ Display graph of the input trees. $(iv)$ Edge label intersection graph of the input trees.}
 \label{fig:example}
\end{figure}

\section{Characterization using minimal cuts}
Let $\Prof = \{T_1, T_2, \cdots, T_k\}$ be a profile of phylogenetic trees. A minimal separator $F$ of $L(\Prof)$ is \emph{legal} if for every $F' \subseteq F$ such that $F' \subseteq V(L(T))$ for some input tree $T$, $F'$ is a clique in $L(T)$. 

\begin{theorem}\label{thm:edge_label}
Profile $\Prof$ is compatible if and only if there exists a maximal set $\F$ of pairwise parallel minimal separators in $L(\Prof)$ where every separator in $\F$ is legal.
\end{theorem}
\begin{proof}
Assume that $\Prof$ is compatible. From Theorem~\ref{thm:EL_compatibility}, there exists a restricted triangulation $H$ of $L(\Prof)$. We can assume that $H$ is minimal triangulation, since if it is not, a restricted minimal triangulation of $L(\Prof)$ can be obtained by repeatedly deleting fill-in edges from $H$ until it is a minimal triangulation. Let $\F = \triangle_{H}$. From Theorem~\ref{thm:parallel_minimal_ct}, $\F$ is a maximal set of pairwise parallel minimal separators of $L(\Prof)$ and $L(\Prof)_{\F} = H$. Assume that $\F$ contains a separator $F$ that is not legal. Let $\{e, e'\} \subseteq F$ where $\{e, e'\} \subseteq E(T)$ for some input tree $T$ and $e \cap e' = \emptyset$. Vertices of $F$ form a clique in $H$. Thus, $H$ contains  the edge $\{e, e'\}$. Since $\{e, e'\}$ is not a valid edge, $H$ is not a restricted triangulation, which is a contradiction. Hence, every separator in $\F$ is legal.

Let $\F$ be a maximal set of pairwise parallel minimal separators of $L(\Prof)$ where every separator in $\F$ is legal. From Theorem~\ref{thm:parallel_minimal_ct}, $L(\Prof)_{\F}$ is a minimal triangulation of $L(\F)$. If $\{e, e'\} \in E(L(\Prof)_{\F})$ is a fill-in edge, then $e \cap e' = \emptyset$ and there exists a minimal separator $F \in \F$ where $\{e, e'\} \subseteq F$. Since separator $F$ is legal, if $\{e, e'\} \subseteq E(T)$ for some input tree $T$ then $e \cap e' \neq \emptyset$. Thus, both $e$ and $e'$ are not from $L(T)$ for any input tree $T$. Hence, every fill-in edge in $L(\Prof)_{\F}$ is valid, and $L(\Prof)_{\F}$ is a restricted triangulation.
\end{proof}

For any vertex  $u$ of an input tree, $\hat{K}(u)$ represents the set of all vertices of $L(\Prof)$ where for every $e \in \hat{K}(u)$, $u \in e$. 

\begin{lemma}\label{lm:no_loner}
Let $F$ be any minimal separator of $L(\Prof)$ and $u$ be any vertex of any input tree. Then, $\hat{K}(u) \not \subseteq F$.
\end{lemma}
\begin{proof}
Suppose $F$ is a minimal $a$-$b$ separator of $L(\Prof)$ and $u$ is a vertex of some input tree such that $\hat{K}(u) \subseteq F$. Consider any vertex $e \in \hat{K}(u)$. Then, there exists a path $\pi$ from $a$ to $b$ in $L(\Prof)$ where $e$ is the only vertex of $F$ in $\pi$. If such a path $\pi$ did not exist, then $F-e$ would still be a $a$-$b$ separator, and $F$ would not be minimal, a contradiction. Let $e_1$ and $e_2$ be the neighbors of $e$ in $\pi$ and let $e = \{u, v\}$. Since $\hat{K}(u) \subseteq F$, $\pi$ does not contain any other vertex $e'$ where $u \in e'$.  Thus, $e \cap e_1 = v$ and $e \cap e_2 =v$. Let $\pi = a, \cdots, e_1, e, e_2, \cdots, b$. Then $\pi' = a, \cdots, e_1, e_2, \cdots, b$ is also a path from $a$ to $b$. Since $F$ is a separator, there should exist a vertex in path $\pi'$ which is also in $F$. Thus, there exists two vertices of $F$ in path $\pi$, contradicting the assumption that $\pi$ contains exactly one vertex of $F$. Thus, neither such a minimal separator $F$ nor such a vertex $u$ can exist.
\end{proof}

A cut $F$ of display graph $G(\Prof)$ is \emph{legal}, if it satisfies the following:
\begin{enumerate}[(LC1)]
\item For every tree $T \in \Prof$, the edges of $T$ in $F$ are incident on a common vertex. 
\item There is at least one edge in each of the connected components of $G(\Prof)-F$.
\end{enumerate}

We make use of the following simple observation in some of the proofs in this section.

\begin{observation}\label{obs:path}
Let $I$ be a set of edges of $G(\Prof)$. Then, there exists a path $v_1, v_2, \cdots, v_m$ in $G(\Prof)-I$ where $m \geq 2$,  if and only if there exists a path $\{v_1, v_2\}, \cdots, \{v_{m-1}, v_m\}$ in $L(\Prof)-I$.
\end{observation}

\begin{theorem}\label{thm:cuts_seps}
A set $F$ is a legal minimal separator of $L(\Prof)$ if and only if $F$ is a legal minimal cut of $G(\Prof)$.
\end{theorem}
\begin{proof}
We will prove that if $F$ is a legal minimal separator of $L(\Prof)$ then $F$ is a legal minimal cut of $G(\Prof)$. The proof for the other direction is similar and is omitted. 

First, we show that $F$ is a cut of $G(\Prof)$. Assume the contrary. Let $\{u, v\}$ and $\{p, q\}$ be vertices in different components of $L(\Prof)-F$. Since, $G(\Prof)-F$ is connected, there exists a path between vertices $u$ and $q$. Also, $\{u,v\} \notin F$ and $\{p,q\} \notin F$. Thus, by Observation~\ref{obs:path} there also exists a path between vertices $\{u, v\}$ and $\{p, q\}$ of $L(\Prof)-F$ and are in the same connected component of $L(\Prof)-F$ which is a contradiction. Thus $F$ is a cut of $G(\Prof)$.
 
Next we show that separator $F$ of $L(\Prof)$ is a legal cut of $G(\Prof)$. For every $T \in \Prof$ all the vertices of $L(T)$ in $F$ form a clique in $L(T)$. Thus, all the edges of $T$ in $F$ are incident on a common vertex. Assume that $G(\Prof)-F$ has a connected component with no edge and let $u$ be the vertex in one such component. Then, $\hat{K}(u) \subseteq F$. But, $F$ is a minimal separator of $L(\Prof)$ and by Lemma~\ref{lm:no_loner}, $\hat{K}(u) \not \subseteq F$ which is a contradiction. Thus, $F$ is a legal cut of $G(\Prof)$.

Lastly, we show that $F$ is a minimal cut of $G(\Prof)$. Assume the contrary. Then, there exists $F' \subset F$ where $G(\Prof)-F'$ is disconnected but $L(\Prof)-F'$ is connected. Since $F' \subset F$ and every connected component of $G(\Prof)-F$ has at least one edge, every connected component  of $G(\Prof)-F'$ also has at least one edge. Let $\{u,v\}$ and $\{p, q\}$ be the edges in different components of $G(\Prof)-F'$. Since, $L(\Prof)-F'$ is connected, there exists an path between $\{u,v\}$ and $\{p, q\}$ in $L(\Prof)-F'$. Then, by Observation~\ref{obs:path} there also exists a path between vertices $u$ and $p$ in $G(\Prof)-F'$. Hence, edges $\{u, v\}$ and $\{p, q\}$ are in the same connected component of $G-F'$ which is contradiction. Thus, $F$ is also minimal cut of $G(\Prof)$.
\end{proof}

\begin{lemma}\label{lm:parallel}
Two legal minimal separators $F$ and $F'$ of $L(\Prof)$ are parallel if and only the if legal minimal cuts $F$ and $F'$ are parallel in $G(\Prof)$.
\end{lemma}
\begin{proof}
Assume that legal minimal separators $F$ and $F'$ of $L(\Prof)$ are parallel, but legal minimal cuts $F$ and $F'$ of $G(\Prof)$ are not. Then, there exist $\{\{u, v\}, \{p, q\}\} \subseteq F'$ where $\{u, v\}$ and $\{p, q\}$ are present in different components of $G(\Prof)-F$. Since $F$ and $F'$ are parallel separators in $L(\Prof)$, and $F$ does not contain $\{u, v\}$ and $\{p, q\}$, there exists a path between vertices $\{u, v\}$ and $\{p, q\}$ in $L(\Prof)-F$. Then, by Observation~\ref{obs:path} there also exists a path between vertices $u$ and $q$ in $G(\Prof)-F$. Thus, edges $\{u, v\}$ and $\{p, q\}$ are in the same connected component of $G(\Prof)-F$ which is a contradiction.

The other direction can be proved similarly using Observation~\ref{obs:path}.
\end{proof}

The next lemma follows from the definition of restricted triangulation and is from~\cite{Gysel2012}.

\begin{lemma}\label{lm:two_cliques}
Let $H$ be a restricted triangulation of $L(\Prof)$ and let $(T, B)$ be a clique tree of $H$. Then, for every vertex $e = \{u, v\}$ in $L(\Prof)$, there does not exist a node $x \in V(T)$ where $B(x)$ contains vertices from both $\hat{K}(u)\setminus e$ and $\hat{K}(v)\setminus e$.
\end{lemma}

A set $\F$ of pairwise parallel legal minimal cuts of $G(\Prof)$ is \emph{complete}, if for every input tree $T \in \Prof$ and for every internal edge $e$ of $T$, there exists a cut $F \in \F$ where $e$ is the only edge of $T$ in $F$. 

\begin{example}
For the display graph $G(\Prof)$ of Fig.~\ref{fig:example}, let $\F = \{F_1,F_2,F_3,F_4\}$, where $F_1 = \{\{1,2\}, \{5,6\}\}$, $F_2 = \{\{2,3\}, \{6,7 \}\{5,6\}\}$, $F_3=\{\{4,5\}, \{1,2\},\{1,c\}\}$ and $F_4=\{\{6,7\},\{2,f\}\}$. Then, $\F$ is a complete set of pairwise parallel legal minimal cuts.
\end{example}

\begin{lemma}\label{lm:maximal_seps_equals_maximal_cuts}
If there exists a maximal set $\F$ of pairwise parallel minimal separators of $L(\Prof)$ where every separator $F \in \F$ is legal, then there exists a complete set of pairwise parallel legal minimal cuts for $G(\Prof)$.
\end{lemma}
\begin{proof}
We will show that for every internal edge $e=\{u,v\}$ of an input tree $T$ there exists a minimal separator in $\F$ which contains only vertex $e$ from $L(T)$. Then from Theorem~\ref{thm:cuts_seps} and Lemma~\ref{lm:parallel} it follows that $\F$ is a complete set of pairwise parallel legal minimal cuts for display graph $G(\Prof)$.

As shown in proof of Theorem~\ref{thm:edge_label}, $L(\Prof)_{\F}$ is a restricted minimal triangulation of $L(\Prof)$. Let $(S, B)$ be a clique tree of $L(\Prof)_{\F}$. By definition, the vertices in each of the sets $\hat{K}(u)$ and $\hat{K}(v)$ form a clique in $L(\Prof)$. Consider any vertex $p$ of $S$ where $\hat{K}(u) \subseteq B(p)$ and any vertex $q$ of $S$ where $\hat{K}(v) \subseteq B(q)$. Since $(S, B)$ is a clique tree of $L(\Prof)_{\F}$, there will always exist such vertices $p$ and $q$. Also, by Lemma~\ref{lm:two_cliques}, $p \neq q$, $B(p) \cap \{\hat{K}(v) \setminus e\} = \emptyset$ and $B(q) \cap \{\hat{K}(u) \setminus e\} = \emptyset$

Let $\pi = p, x_1, x_2, \cdots, x_m, q$ be the path from $p$ to $q$ in $S$ where $m \geq 0$. Let $x_0=p$ and $x_{m+1}=q$. Let $x_i$ be the vertex nearest to $p$ in path $\pi$ where $i \in [m+1]$ and $B(x_i) \cap \{\hat{K}(u) \setminus e\} = \emptyset$. Let $F = B(x_{i-1}) \cap B(x_i)$. Then by Theorem~\ref{thm:parallel_minimal_ct}, $F \in \F$. Since $\hat{K}(u) \cap \hat{K}(v) = e$, by the coherence property of the clique tree, $e \in B(x_j)$ for every $j \in [m]$. Thus, $e \in F$. By Lemma~\ref{lm:two_cliques}, $B(x_{i-1}) \cap \{\hat{K}(v) \setminus e\} = \emptyset$. Since $B(x_i) \cap \{\hat{K}(u) \setminus e\} = \emptyset$, $F \cap \hat{K}(u) = e$ and $F \cap \hat{K}(v) = e$. Thus, for every vertex $e' \in L(T)$ where $e \neq e'$ and $e \cap e' \neq \emptyset$, $e' \notin F$. Also, since every separator in $\F$ is legal, for every vertex $f \in L(T)$ where $f \cap e = \emptyset$, $f \notin F$. Thus, $e$ is the only vertex of $L(T)$ in $F$.
\end{proof}

\begin{lemma}\label{lm:maximal_cuts_equals_maximal_seps}

If there exists a complete set of pairwise parallel legal minimal cuts for $G(\Prof)$, then there exists a maximal set $\F$ of pairwise parallel minimal separators of $L(\Prof)$ where every separator in $\F$ is legal.
\end{lemma}
\begin{proof}
Consider any complete set of pairwise parallel legal minimal cuts $\F'$ of $G(\Prof)$. By Theorem~\ref{thm:cuts_seps} and Lemma~\ref{lm:parallel}, $\F'$ is a set of pairwise parallel legal minimal separators of $L(\Prof)$. There exists a maximal set $\F$ of pairwise parallel minimal separators where $\F' \subseteq \F$. Assume that there exists a minimal separator $F$ in $\F \setminus \F'$ which is not legal. Using an argument similar to the proof of Theorem~\ref{thm:cuts_seps}, it can be shown that $F$ is a cut of $G(\Prof)$, every connected component of $G(\Prof)-F$ has at least one edge, and $F$ is also a minimal cut of $G$.

Since, by assumption, minimal separator $F$ of $L(\Prof)$ is not legal, there exists a tree $T \in \Prof$ where at least two nonincident edges of $T$ are in $F$. Let $e_1=\{x, y\}$ and $e_2=\{x', y'\}$ be those nonincident edges. Consider the any internal edge $e_3$ in $T$ where $e_1$ and $e_2$ are in different components of $T-e_3$. Such an edge exists because $e_1$ and $e_2$ are nonincident. Set $\F'$ is a complete set of pairwise parallel legal minimal cuts of $G$. Thus, there exists minimal cut $F' \in \F'$ where $e_3$ is the only edge of $T$ in $F'$. Since, minimal separators $F$ and $F'$ are in $\F$, they are parallel to each other and vertices $e_1$ and $e_2$ are in the same connected component of in $L(\Prof)-F'$. Thus, by Observation~\ref{obs:path}, there exists a path between vertices $x$ and $x'$ in $G(\Prof)-F'$ and edges $e_1$ and $e_2$ are also in the same connected component of $G(\Prof)-F'$. But that is impossible, since, $F'$ is a legal minimal cut of $G$ and $e_1$ and $e_2$ are in different components of $T-e_3$.

Thus, every separator of $\F \setminus \F'$ is legal and the set $\F$ is a maximal set of pairwise minimal separators of $L(\Prof)$ where every separator in $\F$ is legal.
\end{proof}

Lemmas~\ref{lm:maximal_seps_equals_maximal_cuts} and~\ref{lm:maximal_cuts_equals_maximal_seps} imply the following theorem.

\begin{theorem}\label{thm:maximal_seps_equals_maximal_cuts}
 There exists a maximal set $\F$ of pairwise parallel minimal separators of $L(\Prof)$ where every separator in $\F$ is legal if and only if there exists a complete set of pairwise parallel minimal cuts for $G(\Prof)$.
\end{theorem}

The next result follows from Theorems~\ref{thm:edge_label} and~\ref{thm:maximal_seps_equals_maximal_cuts}.

\begin{theorem}\label{thm:cuts_equal_compatibility}
A profile $\Prof$ of unrooted phylogenetic trees is compatible if and only if there exists a complete set of pairwise parallel legal minimal cuts for $G(\Prof)$.
\end{theorem}

An analogue of Theorem~\ref{thm:cuts_equal_compatibility} can be derived for the edge label intersection graph $L(\Prof)$ as follows. A set $\F$ of legal minimal separators of $L(\Prof)$ is \emph{complete}, if for every internal edge $e$ of an input tree $T$, there exists a separator $F \in \F$ where $e$ is the only vertex of $L(T)$ in $F$. The next theorem follows from Theorems~\ref{thm:cuts_seps} and~\ref{thm:cuts_equal_compatibility}, and Lemma~\ref{lm:parallel}.

\begin{theorem}
A profile $\Prof$ of unrooted phylogenetic trees is compatible if and only if there exists a complete set of pairwise parallel legal minimal separators for $L(\Prof)$.
\end{theorem}

\section{Relationship to splits compatibility}
A \emph{split} of a label set $L$ is a bipartition of $L$. We denote a split $\{X, Y\}$ by $X|Y$. Let $T$ be a phylogenetic tree. Consider an internal edge $e$ of $T$. Deletion of $e$ breaks $T$ into two subtrees $T_1$ and $T_2$. Let $L_1$ and $L_2$ denote the set of all labels in $T_1$ and $T_2$ respectively. Set $\{L_1, L_2\}$ is a split of $\Leaves(T)$. We denote the split corresponding to edge $e$ of $T$ by $\Sigma_e(T)$ and we denote by $\Sigma(T)$ the set of all splits corresponding to all internal edges of $T$. 

A tree $T$ \emph{displays} a split $X$ if there exists an internal edge $e$ of $T$ where $\Sigma_e(T) = X$. Then, we also say $T$ is \emph{compatible} with $X$. A set of splits is compatible if there exists a tree which displays all the splits in the set. Two splits $A_1|A_2$ and $B_1|B_2$ are compatible if and only if at least one of $A_1 \cap B_1$, $A_1 \cap B_2$, $A_2 \cap B_1$ and $A_2 \cap B_2$ is empty~\cite{SempleSteel03}. By the Splits Equivalence Theorem~\cite{Buneman71,SempleSteel03}, a collection of splits is compatible if and only if every pair is compatible.

\begin{lemma}
Let $F$ be a legal minimal cut of $G(\Prof)$ and let $G_1$ and $G_2$ be the two connected components of $G(\Prof)-F$. Then, $\Leaves(G_1)| \Leaves(G_2)$ is a split of $\Leaves(\Prof)$. 
\end{lemma}
\begin{proof}
Consider $G_i$ for any $i \in \{1,2\}$. We will show that $\Leaves(G_i)$ is non-empty. Since $F$ is a legal minimal cut, $G_i$ contains at least one edge $e$ of $G(\Prof)$. If $e$ is a non internal edge, then $\Leaves(G_i)$ is non-empty.  Assume that $e=\{u,v\}$ is an internal edge of some input tree $T$. If $F$ does not contain an edge of $T$, then $\Leaves(T) \subseteq \Leaves(G_i)$ and thus $\Leaves(G_i)$ is non empty. Assume that $F$ contains one or more edges of $T$. Let $T_u$, $T_v$ be the two subtrees of $T-e$. Since $F$ is a legal minimal cut, $F$ contains edges from either $T_u$ or $T_v$ but not both. Without loss of generality assume that $F$ does not contain edges from $T_u$. Then, every edge of $T_u$ will be in the same component as $e$. Since $T_u$ contains at least one leaf vertex, $\Leaves(G_i)$ is non-empty.  
Thus, $\Leaves(G_1)|\Leaves(G_2)$ is a split of $\Leaves(\Prof)$.
\end{proof}

For any legal minimal cut $F$ of $G(\Prof)$, we denote by $\Sigma(F)$ the split of $\Leaves(\Prof)$ induced by $F$.  If $\F$ is a set of legal minimal cuts of $G(\Prof)$, then we denote $\bigcup_{F \in \F}\Sigma(F)$ by $\Sigma(\F)$.

\begin{lemma}\label{lm:parallel_implies_compatible}
Let $F_1$ and $F_2$ be two parallel legal minimal cuts of $G(\Prof)$. Then, $\Sigma(F_1)$ and $\Sigma(F_2)$ are compatible.
\end{lemma}
\begin{proof}
Let $\Sigma(F_1) = U_1| U_2$ and $\Sigma(F_2) = V_1|V_2$. Assume that $\Sigma(F_1)$ and $\Sigma(F_2)$ are incompatible. Thus, the intersection of $U_i$ and $V_j$ for every $i \in \{1,2\}$ and $j \in \{1,2\}$ is non-empty. Let $a \in U_1 \cap V_1$, $b \in U_1 \cap V_2$, $c \in U_2 \cap V_1$ and $d \in U_2 \cap V_2$.  Since $\{a, b\} \subseteq U_1$, there exists a path $\pi_1$ between leaf vertices $a$ and $b$ in $G(\Prof)-F_1$. But $a$ and $b$ are in different components of $G(\Prof)-F_2$. Thus, an edge $e_1$ in path $\pi_1$ in $G(\Prof)$ is in the cut $F_2$. Similarly, $\{c,d\} \subseteq U_2$ and there exists a path $\pi_2$ between labels $c$ and $d$ in $G(\Prof)-F_1$. Since $c$ and $d$ are in different components of $G(\Prof)-F_2$, cut $F_2$ contains an edge $e_2$ in path $\pi_2$. But paths $\pi_1$ and $\pi_2$ are in different components of $G(\Prof)-F_1$. So, edges $e_1$ and $e_2$ are in different components of $G(\Prof)-F_1$. Since $\{e_1, e_2\} \subseteq F_2$, the cuts $F_1$ and $F_2$ are not parallel, which is a contradiction.
\end{proof}

\begin{lemma}\label{lm:splits_compatibility}
Let $\F$ be a complete set of pairwise parallel legal minimal cuts of the display graph of profile $\Prof$. The following statements hold.
\begin{enumerate}[(i)]
\item $\Sigma(\F)$ is compatible.
\item If $S$ is compatible tree for $\Sigma(\F)$, then $S$ is a compatible tree for $\Prof$.
\item There exists a compatible tree $S$ of $\Prof$ where $\Sigma(S) = \Sigma(\F)$.
\end{enumerate}
\end{lemma}
\begin{proof}

(i) The statement follows from Lemma~\ref{lm:parallel_implies_compatible} and the splits equivalence theorem.

(ii) Let $T$ be an input tree of $\Prof$ and let $S' = S|\Leaves(T)$. We will show that $S'$ displays $\Sigma(e)$ for every internal edge $e$ of $T$. Let $\Sigma(e) = A|B$. There exists a cut $F \in \F$ where $e$ is the only edge of $T$ in $F$. Since $F$ is a minimal cut, the endpoints of $e$ are in different components of $G(\Prof)-F$. Thus, if $\Sigma(F) = A'|B'$ then up to relabeling of sets we have $A \subseteq A'$ and $B \subseteq B'$. Because $S$ displays $\Sigma(F)$, $S'$ also displays $\Sigma(e)$. Since $S'$ displays all the splits of $T$, $T$ can be obtained from $S'$ by contraction of zero or more edges~\cite{SempleSteel03}.  Thus, $S$ displays $T$. Since $S$ displays every tree in $\Prof$, $S$ is a compatible tree of $\Prof$.

(iii) This is a consequence of the well-known fact (see, e.g.,~\cite{SempleSteel03}) that if $X$ is a set of compatible splits, there exists a tree $T$ where $\Sigma(T) = X$.
\end{proof}

\begin{example}
For the cuts of the display graph in Fig.~\ref{fig:example} given in Example 1, we have $\Sigma(F_1) =abc | defg$, $\Sigma(F_2) = abcfg | de$, $\Sigma(F_3) = ab | cdefg$, and $\Sigma(F_4) = abcde | fg$.  Note that these splits are pairwise compatible.
\end{example}

\section{Relationship to legal triangulation}
Let $\Prof$ be a profile of phylogenetic trees. Theorems~\ref{thm:lt} and~\ref{thm:cuts_equal_compatibility} together imply that if $G(\Prof)$ has a complete set $\F$ of pairwise parallel legal minimal cuts, there also exists a legal triangulation of $G(\Prof)$. As shown in~\cite{Vakati11}, a legal triangulation of $G(\Prof)$ can be derived from a compatible tree of $\Prof$. In this section, we show how to derive a legal triangulation of $G(\Prof)$ directly from $\F$ without building a compatible tree. This shows the relationship between complete sets of pairwise parallel legal minimal cuts and legal triangulations of display graphs. By Theorems~\ref{thm:edge_label} and ~\ref{thm:cuts_seps} and Lemma~\ref{lm:parallel}, this also shows the relationship between restricted triangulations of edge label intersection graphs and legal triangulations of display graphs.

A complete set $\F$ of pairwise parallel legal minimal cuts of $G(\Prof)$ is \emph{minimal} if no proper subset of $\F$ is also complete. Let $\F$ be a minimal complete set of pairwise parallel legal minimal cuts of $G(\Prof)$. 

At a high level, we construct a legal triangulation of $G(\Prof)$ from $\F$ as follows. Consider any cut $F \in \F$. We build a pair $D_F=(X,Y)$ where $X$ and $Y$ are subsets of $E(F)$ and are vertex separators of $G(\Prof)$. Let $A$ and $B$ be the connected components of $G(\Prof)-F$. Also, let $A'$, $B'$ be the subgraphs induced in $G(\Prof)$ by the vertex sets $V(A) \cup \{X \cap Y\}$ and $V(B) \cup \{X \cap Y\}$ respectively. To legally triangulate $G(\Prof)$ we first triangulate the subgraph of $G(\Prof)$ induced by the vertex set $X \cup Y$ and then triangulate the subgraphs $A'$ and $B'$. To triangulate either of those subgraphs, we again use vertex separators built from endpoints of a different cut. We make sure that, for every set $D_I$ for some $I \in \F$ built after $D_F$, both the sets of $D_I$ are subsets of either $V(A')$ or $V(B')$ but not both.

We now give the details of our construction. We consider the elements of $\F$ in some arbitrary, but fixed order, and use a set $W$ to record all such cuts $F \in \F$  for which $D_F$ has already been constructed. Initially $W$ is empty. For each successive cut $F$ in $\F$, we do the following. Let $F' \subseteq F$ be the set of all internal edges $e \in F$ such that $e$ is the only edge of the tree containing $e$ that is in $F$.  Let $A$ and $B$ be the two connected components of $G(\Prof)-F$. Let $X =  V(A) \cap V(F')$ and $Y= V(B) \cap V(F')$. For every edge $e $ of $F'$ whose endpoints are in different sets of some set $D_I$ where $I \in W$, we do the following.  Let $Q$ be the connected component of $G(\Prof)-I$ where $E(Q) \cap F \neq \emptyset$. Note that $Q$ is the only such component of $G(\Prof)-I$. Let $v$ be the vertex of $e$ in $Q$. Replace the endpoints of $e$ in sets $X$ and $Y$ by $v$. For every non internal edge $f \in F$ where $f$ is the only edge of the tree containing $f$ that is in $F$, add the internal vertex of $f$ to both sets $X$ and $Y$. If there exists a tree $T$ where more than one edge of $T$ is in $F$, add the common endpoint of all the edges of $T$ in $F$ to both sets $X$ and $Y$. Set $D_F$ to $(X, Y)$. Add $F$ to $W$. 

For every $F \in \F$, let $O_F$ be the set defined as follows.  Let $D_F=(X,Y)$ and let $X=\{x_1, \cdots, x_m, z_1,\cdots, z_p\}$ and $Y =  \{y_1, \cdots, y_m, z_1, \cdots, z_p\}$, where $m > 0$, $p \geq 0$ and for every $i \in [m]$, $\{x_i, y_i\}$ is an internal edge of $G(\Prof)$. Then, $O_F$ consists of sets $\{x_1, \cdots, x_j, y_j, \cdots, y_m, z_1, \cdots, z_p\}$ for every $j \in [m]$.

Let $G'$ be the graph derived from $G(\Prof)$ as follows. For every cut $F \in \F$ where $D_F=(X,Y)$, add edges to make each of the sets $X$ and $Y$ a clique. For every cut $F \in \F$ and for every $Y \in O_F$, add edges to make $Y$ a clique. For every leaf $\ell$, make the vertices of $N_{G(\Prof)}(\ell)$ a clique.

\begin{theorem}\label{thm:lt_convert}
$G'$ is a legal triangulation of $G(\Prof)$.
\end{theorem}

To prove Theorem~\ref{thm:lt_convert} we first prove few useful lemmas. For every cut $F \in \F$ where $D_F=(X,Y)$, we denote the sets $X \cup Y$, $X \cap Y$ by $F_\cup$ and $F_\cap$ respectively. For any internal edge $e$, we call the cut $F \in \F$ a \emph{differentiating} cut of $e$ if $e$'s endpoints are in different sets of $D_F$. Note that, since $\F$ is minimal, every cut in $\F$ is a differentiating cut of some internal edge. A clique of $G'$ is \emph{illegal} if it contains a fill-in edge with a leaf vertex as an endpoint or if it contains an internal edge along with any another edge of $G(\Prof)$. Graph $G'$ is a legal triangulation if and only if $G'$ does not contain an illegal clique.

\begin{lemma}\label{lm:helper_lemma}
Let $F$ and $I$ be two distinct cuts of $\F$. Let $x$ be a vertex where $x \in F_\cup$ and $x$ is in the connected component of $G(\Prof)-I$ which does not contain edges of $F$. Then, $x \in I_\cap$.
\end{lemma}
\begin{proof}
Let $E_F$ be the set of all edges of $F$ that have $x$ as an endpoint and let $E_I$ be the set of all edges of $I$ that have $x$ as an endpoint. Since $x$ is in $F_\cup$ and in the component of $G(\Prof)-I$ which does not contain edges of $F$, $E_F \subseteq E_I \subseteq I$. If $|E_I| > 1$, then $x \in I_\cap$. Assume that $|E_I| =1$ and let $e=\{x, y\}$ be the edge with endpoint $x$ in $I$. Since $E_F \subseteq E_I$ and $E_F \geq 1$, $e \in F$ and $|E_F| = 1$.  

If $y$ is a leaf vertex, then $x \in I_\cap$, so assume that $y$ is not a leaf vertex.  Let $E_y$ represent the set of edges of $I$ with $y$ as an endpoint. If $|E_y | > 1$, then $x \notin F_\cup$ since $E_y \subseteq F$. Thus, $|E_y| =1$. Let $J$ be the cut that differentiates edge $e$. If $F = J$ then by construction, $x \in I_\cap$. Thus, assume that $F \neq J$. If $J$ is in the same connected component of $G(\Prof) - F$ as $I$, then by construction $x \notin F_\cup$, which is a contradiction. Thus, $J$ is in the connected component of $G(\Prof)-F$ which does not contain $I$ and by construction, $x \in I_\cap$.
\end{proof}

\begin{lemma}\label{lm:forbidden_edges}
Let $D_F = (X,Y)$ for some $F \in \F$.  Let $A$ and $B$ be the connected components of $G(\Prof)-F$ where $\{X \setminus F_\cap\} \subseteq V(A)$ and $\{Y \setminus F_\cap\} \subseteq V(B)$. There does not exist an edge $\{u, v\}$ in $G'$ where
\begin{enumerate}
\item $u \in V(A) \setminus F_\cap$ and $v \in V(B) \setminus Y$, or
\item $u \in V(B) \setminus F_\cap$ and $v \in V(A) \setminus X$
\end{enumerate}
\end{lemma}
\begin{proof}
Assume that there exists an edge $e=\{u, v\}$ in $G'$ which satisfies one of the two conditions. Without loss of generality, assume that $u \in V(A) \setminus F_\cap$ and $v \in V(B) \setminus Y$.  If $e \in E(G(\Prof))$, then $e \in F$ and hence, by construction, at least one of $u$ and $v$ should be in $F_\cup$. But $v \notin Y$ and so, $u \in F_\cap$ which is a contradiction. Thus, $e$ is a fill-in edge. Note that $e \not \subseteq F_\cup$.  So, by construction there exists a cut $I \in \F$ where $I \neq F$ and $e \subseteq I_\cup$. 

If $E(A) \cap I \neq \emptyset$, then by Lemma~\ref{lm:helper_lemma}, $v \in F_\cap$, which is a contradiction. Thus, assume that $E(B) \cap I \neq \emptyset$. Then, by Lemma~\ref{lm:helper_lemma}, $u \in F_\cap$ which is a contradiction. Thus, such an edge $e$ cannot exist.
\end{proof}

\begin{lemma}\label{lm:vf_legal}
Let $F$ be a cut of $\F$ and let $H$ represent the subgraph of $G'$ induced by vertices of $F_\cup$. Then, we have the following.
\begin{enumerate}
\item  $H$ is triangulated.
\item The is no illegal clique in $H$.
\end{enumerate}
\end{lemma}
\begin{proof}
Let $D_F = (X,Y)$ and let $A$, $B$ be the connected components of $G(\Prof)-F$. Let $X=\{x_1, \cdots, x_m, z_1, \cdots, z_p\}$ and $Y =  \{y_1, \cdots, y_m, z_1, \cdots, z_p\}$ where for every $i \in [m]$, $x_i \in V(A)$, $y_i \in V(B)$ and $\{x_i, y_i\}$ is an internal edge of $G(\Prof)$. 

We will first prove that for any $i \in [m]$, $j \in [m]$ where $i > j$, there does not exist an edge $e=\{x_i, y_j\}$ in $H$. Assume that $e \in E(H)$. Edges $e_1=\{x_i, y_i\}$ and $e_2=\{x_j, y_j\}$ are differentiated by $F$. Thus, $e$ is not in $E(G(\Prof))$ and is a fill-in edge. Since there does not exist a set in $O_F$ which contains both $x_i$ and $y_j$, there exists a cut $I \in \F$ where $e \subseteq I_\cup$. Since $F$ and $I$ are parallel, edges of $I$ are either in component $A$ or $B$ but not both. Assume that $I \cap E(A) \neq \emptyset$. Then by Lemma~\ref{lm:helper_lemma}, $y_j \in F_\cap$, which is a contradiction. Similarly, if $I \cap E(B) \neq \emptyset$, then by Lemma~\ref{lm:helper_lemma}, $x_i \in F_\cap$ which is a contradiction. Thus, there cannot be such a fill-in edge $e$.

Let $C$ be a chordless cycle of length at least four in $H$. Sets $X$ and $Y$ are cliques in $G'$. Thus, if $C$ contains more than two vertices from one of $X$ or $Y$, $C$ must contain a chord. Hence, $C$ has exactly four vertices and contains exactly two vertices each from $X$ and $Y$. Note that $C$ cannot contain vertex $z_i$ for any $i \in [p]$. Let $x_i$,$x_j$ be the vertices of $X$ in $C$ where $1 \leq i < j \leq m$. Similarly, let $y_{i'}$,$y_{j'}$ be the vertices of $Y$ in $C$ where $1 \leq i' < j' \leq m$.  We have the following cases. If $i \leq i'$, then $\{x_1, \cdots, x_{i}, y_{i}, \cdots, y_m, z_1, \cdots, z_p\} \in O_F$ and thus vertices $x_i, y_{i'}, y_{j'}$ form a clique. Hence, $C$ is not chordless which is a contradiction. If $i > i'$, then from the above argument neither of the edges $\{x_i, y_{i'}\}$ and $\{x_j, y_{i'}\}$ can exist. Thus, vertex $y_{i'}$ cannot be in $C$ which is a contradiction.

Assume that $H$ contains an illegal clique $H'$. Thus, $H'$ contains two internal edges $e$ and $e'$. By construction, $F_\cup$ cannot contain a leaf vertex. By legality of cuts and from the construction of $F_\cup$, edges $e$ and $e'$ are from different input trees and both are differentiated by $F$.  Let $e = \{x_i, y_i\}$ for some $i \in [m]$ and let $e' = \{x_j, y_j\}$ for some $j \in [m]$. Without loss of generality, assume that $i < j$. As proved above, there cannot exist an edge between vertices $x_j$ and $y_i$ in $H$ and thus $H'$ is not a clique which is a contradiction. Thus, $H$ does not contain an illegal clique.
\end{proof}

\begin{lemma}\label{lm:C_in_F}
G' is chordal.
\end{lemma}
\begin{proof}
Assume the contrary. Let $C$ be a chordless cycle of length at least 4 in $G'$. By construction, $C$ cannot contain a leaf vertex. We have the following cases. 

Suppose that there exist vertices $\{u, v\} \subset V(C)$ and a cut $F \in \F$ where if $D_F = (X,Y)$, then $u \in X \setminus F_\cap$ and $v \in Y \setminus F_\cap$. Let $A$, $B$ be the connected components of $G(\Prof)-F$ where $u \in V(A)$ and $v \in V(B)$. We have two cases. 

\begin{enumerate}[{Case} 1:]
\item $C$ contains a vertex $x \in F_{\cap}$. Then, there exists a path $u, x, v$ in $C$.    Because $C$ is a cycle, there must exist an edge between a vertex $u' \in V(A) \setminus {x}$ and $v' \in V(B) \setminus {x}$. Since $C$ is chordless, $u' \notin F_\cap$ and $v' \notin F_\cap$. Thus, $u' \in V(A) \setminus F_\cap$ and $v' \in V(B) \setminus F_\cap$. By Lemma~\ref{lm:forbidden_edges}, if $u' \in V(A) \setminus X$ then there cannot exist an edge between $u'$ and $v'$. Thus, $u' \in X \setminus F_\cap$. Similarly, $v' \in Y \setminus F_\cap$. If $u \neq u'$ or $v \neq v'$, $C$ cannot be chordless. Thus, $u = u'$ and $v = v'$ and $C$ is a chordless cycle of length 3 which is a contradiction.

\item $C$ does not contain a vertex of $F_{\cap}$. Since $u \in V(A)\setminus F_{\cap}$, $v \in V(B) \setminus F_{\cap}$ and $F$ is a cut, there must exist two edges $e_1= \{x_1, y_1\}$ and $e_2=\{x_2, y_2\}$ in $C$ where $\{x_1, x_2\} \subseteq V(A) \setminus F_{\cap}$ and $\{y_1, y_2\} \subseteq V(B) \setminus F_{\cap}$. If $x_1 \in V(A) \setminus X$, then by Lemma~\ref{lm:forbidden_edges} there cannot exist an edge between $x_1$ and $y_1$. Thus, $x_1 \in X \setminus F_\cap$. Similarly, $x_2 \in X \setminus F_\cap$ and $\{y_1, y_2\} \subseteq Y \setminus F_\cap$. Since vertex sets of $X$ and $Y$ are cliques in $G'$, there exist edges $\{x_1, x_2\}$ and $\{y_1, y_2\}$. Thus, there cannot exist any other vertex in $C$ and hence $V(C) \subseteq F_\cup$. But, by Lemma~\ref{lm:vf_legal} subgraph of $G'$ induced by vertices of $F_\cup$ is triangulated. Thus, $C$ is not chordless which is a contradiction.
\end{enumerate}

Now assume that for every cut $F \in \F$, where $D_F = (X,Y)$, there do not exist two vertices $u, v \in V(C)$ where $u \in X \setminus F_\cap$ and $v \in Y \setminus F_\cap$. This also implies that, for every cut $F \in \F$ at most two vertices of $V(C)$ are in $F_\cup$. Let $x_1, x_2, x_3, x_4$ be a path of length four in $C$. Also, for every $i \in \{1,2,3\}$, let $F^{(i)} \in \F$ be the cut where $\{x_i, x_{i+1}\} \subseteq F^{(i)}_\cup$. Note that such cuts must exist and must be distinct. For every $i \in \{1,2,3\}$, let $A_i$ and $B_i$ be the connected components of $G(\Prof)-F^{(i)}$. Without loss of generality, assume that $E(A_1) \cap F^{(2)} \neq \emptyset$ and $E(B_2) \cap F^{(1)} \neq \emptyset$. We have the following cases.

\begin{enumerate}[{Case} 1:]
\item $F^{(3)} \cap E(A_2) \neq \emptyset$. If $x_1 \in A_2$, then by Lemma~\ref{lm:helper_lemma}, $x_1 \in F^{(2)}_\cap$ and $C$ is not chordless, which is a contradiction. Thus, $x_1 \in B_2$. Similarly, if $x_4 \in B_2$, by Lemma~\ref{lm:helper_lemma}, $x_4 \in F^{(2)}_\cap$ and $C$ is not chordless, which is a contradiction. Thus, $x_4 \in A_2$.  Since $C$ is a cycle, $F^{(2)}$ is a minimal cut and $\{F^{(2)}_\cup \setminus \{x_2, x_3\}\} \cap V(C) = \emptyset$, there exists an edge $\{v_1, v_2\}$ in $C$ where $v_1 \in V(A_2) \setminus F^{(2)}_\cup$ and $v_2 \in V(B_2) \setminus F^{(2)}_\cup$. But, by Lemma~\ref{lm:forbidden_edges}, such an edge cannot exist.

\item $F^{(3)} \cap E(A_1) \neq \emptyset$ and $F^{(3)} \cap E(B_2) \neq \emptyset$. Without loss of generality let $A_3$, $B_3$ be the connected components of $G(\Prof) -F^{(3)}$ that contain $F^{(2)}$ and $F^{(1)}$ respectively. Assume that $x_2 \in A_3$. Since $x_2 \in F^{(1)}_\cup$, by Lemma~\ref{lm:helper_lemma}, $x_2 \in F^{(3)}_\cap$. Then, there exists an edge $\{x_2, x_4\}$ and $C$ is not chordless, which is a contradiction. Thus, $x_2 \in B_3$. But $x_2 \in F^{(2)}_\cup$ and thus, by Lemma~\ref{lm:helper_lemma}, $x_2 \in F^{(3)}_\cap$. Hence, there exists a chord $\{x_2, x_4\}$ and $C$ is not chordless, which again is a contradiction.

\item $F^{(3)} \cap E(B_1) \neq \emptyset$. Renaming vertices $x_1$, $x_2$, $x_3$ and $x_4$ as, $x_4$, $x_3$, $x_2$ and $x_1$, respectively, brings us back to case 2.
\end{enumerate} 
Thus, $G'$ does not contain a chordless cycle of length 4 or greater and hence $G'$ is chordal.
\end{proof}

\begin{proof}[Proof of Theorem~\ref{thm:lt_convert}]
From Lemma~\ref{lm:C_in_F}, $G'$ is triangulated. We now prove that triangulation $G'$ is legal.

By construction, we do not add any fill-in edge with a leaf vertex as an endpoint. Thus, condition $(LT2)$ is true for $G'$. Assume that there exists a clique $H$ with two internal edges $e=\{x_1, y_1\}$ and $e'=\{x_2, y_2\}$. Let $F$ be the cut which differentiates $e$. Let $A$ and $B$ be the connected components of $G(\Prof)-F$ where $x_1 \in V(A)$ and $y_1 \in V(B)$. By Lemma~\ref{lm:vf_legal}, both the endpoints of $e'$ cannot be in $F_\cup$. Without loss of generality, assume that $x_2 \notin F_\cup$ and $x_2 \in A$. Since $x_2 \notin F_\cup$ and $y_1 \notin F_\cap$, by Lemma~\ref{lm:forbidden_edges}, there cannot exist an edge between $x_2$ and $y_1$ in $G'$. Thus, $H$ is not a clique of $G'$ which is a contradiction.
\end{proof}

\section{Conclusion}
We have shown that the characterization of tree compatibility in terms of restricted triangulations of the edge label intersection graph transforms into a characterization in terms of minimal cuts in the display graph. We have also shown how these two characterizations relate to the characterization in terms of legal triangulations of the display graph~\cite{Vakati11}.

It remains to be seen whether any of these characterizations can be exploited to derive an explicit fixed parameter algorithm for the tree compatibility problem when parametrized by number of trees. Grunewald et al.~\cite{Grunewald2008} use quartet graphs to characterize when a collection of quartets define and identify a compatible supertree. An interesting question is whether a similar characterization can be derived for collections of phylogenetic trees using display graphs or edge label intersection graphs.
\section*{Acknowledgments}
We thank Sylvain Guillemot for his valuable suggestions and comments.
\bibliographystyle{abbrv}
\bibliography{cuts_modified}

\begin{thebibliography}{10}

\bibitem{Bouchitte2001}
V.~Bouchitt\'{e} and I.~Todinca.
\newblock {Treewidth and minimum fill-in: Grouping the minimal separators}.
\newblock {\em SIAM J. Comput.}, 31(1):212--232, Jan. 2001.

\bibitem{BryantLagergren06}
D.~Bryant and J.~Lagergren.
\newblock Compatibility of unrooted phylogenetic trees is {FPT}.
\newblock {\em Theor. Comput. Sci.}, 351:296--302, 2006.

\bibitem{Buneman71}
P.~Buneman.
\newblock The recovery of trees from measures of dissimilarity.
\newblock In {\em Mathematics in the Archaeological and Historical Sciences},
  pages 387--395. Edinburgh University Press, Edinburgh, 1971.

\bibitem{Grunewald2008}
S.~Grunewald, P.~J. Humphries, and C.~Semple.
\newblock Quartet compatibility {and} the quartet graph.
\newblock {\em Electron. J. Comb.}, 15(1):R103, 2008.

\bibitem{Gysel2012}
R.~Gysel, K.~Stevens, and D.~Gusfield.
\newblock Reducing problems in unrooted tree compatibility to restricted
  triangulations of intersection graphs.
\newblock In {\em WABI}, pages 93--105, 2012.

\bibitem{Heggernes2006}
P.~Heggernes.
\newblock Minimal triangulations of graphs: A survey.
\newblock {\em Discrete Math.}, 306(3):297 -- 317, 2006.

\bibitem{Parra1997}
A.~Parra and P.~Scheffler.
\newblock Characterizations and algorithmic applications of chordal graph
  embeddings.
\newblock {\em Discrete Appl. Math.}, 79(1-3):171--188, 1997.

\bibitem{SempleSteel03}
C.~Semple and M.~Steel.
\newblock {\em Phylogenetics}.
\newblock Oxford Lecture Series in Mathematics. Oxford University Press,
  Oxford, 2003.

\bibitem{Steel92}
M.~A. Steel.
\newblock The complexity of reconstructing trees from qualitative characters
  and subtrees.
\newblock {\em J. Classif.}, 9:91--116, 1992.

\bibitem{Vakati11}
S.~Vakati and D.~Fern{\'a}ndez-Baca.
\newblock Graph triangulations and the compatibility of unrooted phylogenetic
  trees.
\newblock {\em Appl. Math. Lett.}, 24(5):719--723, 2011.

\end{thebibliography}
\end{document}